\documentclass[prl,twocolumn,superscriptaddress]{revtex4-1}

\usepackage{amsfonts}
\usepackage{pifont}
\usepackage{bbding}
\usepackage{mathrsfs}
\usepackage{}
\usepackage{textcomp}
\usepackage{latexsym}
\usepackage{amsthm}
\usepackage{amssymb}
\usepackage{amsmath}
\usepackage[all]{xy}
\usepackage{graphicx}
\usepackage{dcolumn}
\usepackage{bm}
\usepackage{color}
\usepackage[T1]{fontenc}
\usepackage[sc]{mathpazo}
\usepackage{amsmath}
\usepackage{amssymb}
\usepackage{enumerate}
\usepackage{amsthm}

\newtheorem{thm}{Theorem}

\usepackage{mathdots}

\def\be{\begin{equation}}
\def\ee{\end{equation}}
\def\bea{\begin{eqnarray*}}
\def\eea{\end{eqnarray*}}

\def\be{\begin{equation}}
\def\ee{\end{equation}}
\def\bea{\begin{equation*}}
\def\eea{\end{equation*}}
\def\bna{\begin{eqnarray*}}
\def\ena{\end{eqnarray*}}
\def\bn{\begin{eqnarray}}
\def\en{\end{eqnarray}}
\def\bpm{\begin{pmatrix}}
\def\epm{\end{pmatrix}}

\newcommand{\bra}[1]{\langle#1|}

\newcommand{\ket}[1]{|#1\rangle}

\newcommand{\braket}[1]{\langle#1\rangle}

\SelectTips{eu}{11}

\begin{document}

\title{Simulating broken $\cal PT$-symmetric Hamiltonian systems by weak measurement}

 \author{Minyi Huang}
 \email{11335001@zju.edu.cn}
 \affiliation{School of Mathematical Sciences, Zhejiang University, Hangzhou 310027, China}

 \author{Ray-Kuang Lee}
\affiliation{Institute of Photonics Technologies, National Tsing Hua University, Hsinchu 30013, Taiwan}
\affiliation{Center for Quantum Technology, Hsinchu 30013, Taiwan}

\author{Lijian Zhang}
\affiliation{College of Engineering and Applied Sciences, Nanjing University, Nanjing 210093, China}

\author{Shao-Ming Fei}
\affiliation{School of Mathematical Sciences, Capital Normal University, Beijing 100048, China}
\affiliation{Max-Planck-Institute for Mathematics in the Sciences, 04103 Leipzig, Germany}

 \author{Junde Wu}
\affiliation{School of Mathematical Sciences, Zhejiang University, Hangzhou 310027, China}

\begin{abstract}
By embedding a $\cal PT$-symmetric (pseudo-Hermitian) system into a large Hermitian one, we disclose the relations between $\cal{PT}$-symmetric quantum theory and weak measurement theory.
We show that the weak measurement can give rise to the inner product structure of $\cal PT$-symmetric systems, with the pre-selected state and its post-selected state resident in the dilated conventional system. Typically in quantum information theory, by projecting out the irrelevant degrees and projecting onto the subspace, even local broken $\cal PT$-symmetric Hamiltonian systems can be effectively simulated by this weak measurement paradigm.
\end{abstract}

\maketitle

\noindent {\it Introduction}~~
Generalizing the conventional Hermitian quantum mechanics, Bender and his colleagues established the Parity ($\cal P$)-time ($\cal T$)-symmetric quantum mechanics in 1998~\cite{Bender98}.
With the additional degree of freedom from a non-conservative Hamiltonian, as well as the existence of exceptional points between unbroken and broken $\cal PT$-symmetries,
optical $\cal PT$-symmetric devices have been demonstrated with many useful applications~\cite{El-OL, Makris-PRL, Guo-PRL, Ruter-NP, Chang-NP, Tang-NP}.
Although calling for more caution on physical interpretations, especially on the consistency problem of local $\mathcal{PT}$-symmetric operation and the no-signaling principle~\cite{Lee14},
$\cal PT$-symmetric quantum mechanics has been stimulating our understanding on many interesting problems such as
spectral equivalence~\cite{Dorey}, quantum brachistochrone ~\cite{Gunther-PRA} and Riemann hypothesis~\cite{Bender-17}.

Compared with the Dirac inner product in conventional quantum mechanics, $\cal PT$-symmetric quantum theory can be well manifested by the $\eta$-inner product~\cite{Mostafazadeh, Mostafazadeh-Geom}.
In the broken $\cal PT$-symmetry case, the $\eta$-inner product of a state with itself can be negative, which makes the broken $\cal PT$-symmetric quantum systems a complete departure from conventional quantum mechanics.
While in the unbroken $\cal PT$-symmetry case, the $\eta$-inner product presents a completely analogous physical interpretation to the Dirac inner product, giving rise to
many similar properties between $\cal PT$-symmetric and conventional quantum mechanics.
Recent works also show that the $\eta$-inner product is tightly related to the properties of superposition and coherence in conventional quantum mechanics~\cite{JPA}.

Despite the original motivation to build a new framework of quantum theory, researchers are aware of the importance of simulating $\cal PT$-symmetric systems with conventional quantum mechanics. It will help explore the properties and physical meaning of $\cal PT$-symmetric quantum systems.
On this issue, one should answer the question in what sense a quantum system can be viewed as $\cal PT$-symmetric.
One approach, initialized by G{\"u}nther and Samsonov, is to embed
unbroken $\cal PT$-symmetric Hamiltonians into higher dimensional Hermitian Hamiltonians~\cite{Gunther-PRL,Huang,Ueda}.
 By dilating the system to a large Hermitian one and projecting out the ancillary system, this paradigm successfully simulates the evolution of unbroken $\cal PT$-symmetric Hamiltonians. Such a way, inspired by Naimark dilation and typical ideas in quantum simulation, endows direct physical meaning of $\cal PT$-symmetric
 quantum systems in the sense of open systems.
However, the simulation of broken $\cal PT$-symmetric systems is still in suspense, due to its essential distinctions
with conventional quantum systems.

In this Letter, we illustrate the simulation for broken $\cal PT$-symmetric systems based on weak measurement~\cite{AAV}.
For a system weakly coupled to the apparatus, the pointer state will be shifted by the weak value when a weak measurement is performed.
The weak value, tightly related to the non-classical features of quantum mechanics, such as the Hardy's paradox~\cite{Aharonov}, three box paradox~\cite{Resch}  and Leggett-Garg inequalities~\cite{Palacios},
can take values beyond the expected values of an observable, and even be a complex number.
The weak measurement theory has provided new ways to measure geometric phases~\cite{Hosten, Sjogvist, Kobayashi,Lijian} and non-Hermitian systems~\cite{Pati,Vaidman}, as well as
to amplify signals as a sensitive estimation of small evolution parameters~\cite{Lundeen, Starling-2010b, Brunner}.
Our aim is to propose a concrete scenario in which the quantum system can be viewed as $\cal PT$-symmetric by utilizing the weak measurement.
Our result reveals the connections between $\cal PT$-symmetry and the weak measurement theory, providing the important missing point for the simulation 
of broken $\cal {PT}$-symmetric quantum systems.


\noindent {\it Generalized embedding of $\cal PT$-symmetric systems}~~
Consider $n$-dimensional discrete quantum systems. A linear operator $P$ is said to be a parity operator if ${P}^2=I$, where $I$ denotes the $n\times n$ identity matrix.
An anti-linear operator $T$ is said to be a time reversal operator if $T\overline{T}=I$ and $PT=T\overline{P}$,
where $\overline{T}$ ($\overline{P}$) stands for the complex conjugation of $T$ ($P$). A Hamiltonian $H$ is said to be $PT$-symmetric if $HPT=PT\overline{H}$~\cite{note}.
$H$ is called unbroken $\cal PT$-symmetric if it is diagonalizable and all of its eigenvalues are real.
Otherwise, $H$ is called broken $\cal PT$-symmetric.

In quantum mechanics, a Hamiltonian $H$ gives rise to a unitary evolution of the system.
Let $\phi_1$ and $\phi_2$ be two states.
On can introduce a Hermitian operator $\eta$ to define the $\eta$-inner
product by $\braket{\phi_1|\phi_2}_{\eta}=\braket{\phi_1|\eta|\phi_2}$.
With respect to the $\eta$-inner product, $H$ presents a unitary evolution if and only if
${H}^\dag\eta=\eta H$~\cite{Mostafazadeh, Mostafazadeh-Geom, Deng, Mannheim, Horn}, where ${H}^\dag$ denotes the conjugation and transpose of $H$.
Here, $\eta$ is said to be the metric operator of $H$.
Moreover, for a generic $\cal PT$-symmetric operator $H$ and its metric operator $\eta$, there always exist some matrix $\Psi'$ such that
$\Psi'^{-1}H\Psi'=J$ and $\Psi'^\dag\eta\Psi'=S$, where
\begin{equation}\label{cano2}
J=diag(
J_{n_1}(\lambda_1,\overline{\lambda}_1),...,J_{n_p}(\lambda_p,\overline{\lambda}_p),J_{n_{p+1}}(\lambda_{p+1}),...,J_r(\lambda_r)),
\end{equation}
$J_{n_k}(\lambda_k,\overline{\lambda}_k)=
\bpm\begin{smallmatrix} J_{n_k}(\lambda_k)&0\\0&J_{n_k}(\overline{\lambda_k})\end{smallmatrix}\epm$,
$J_{n_j}(\lambda_j)$ are the Jordan blocks, $\lambda_1,
\cdots, \lambda_p$ are complex numbers and $\lambda_{p+1},
\cdots, \lambda_r$ are real numbers,
\begin{equation}
S=
diag(S_{2n_1},...,S_{2n_p},\epsilon_{n_q} S_{n_q},...,\epsilon_{n_r} S_{n_r}),\label{cano2'}
\end{equation}
$n_i$ denote the orders of Jordan blocks in Eq. (\ref{cano2}), i.e.,  $S_{k}=\begin{pmatrix}\begin{smallmatrix}&&1\\&\iddots&\\1&&\end{smallmatrix}\end{pmatrix}_{k\times k}$
and $\epsilon_i=\pm 1$ is uniquely determined by $\eta$ \cite{Huang,Gohberg}.
For convenience, we only consider the situations in which $\epsilon_i=1$.
In this case, $S$ is a permutation matrix and $S^2=I$. Note that $S$ can be equal to $I$ if and only if $H$ is unbroken $\cal PT$-symmetric~\cite{Huang}. Henceforth we always
assume $S=I$ in the unbroken case. The following theorem gives an important property of $\cal PT$-symmetric Hamiltonians.

\begin{thm}\label{thm1}
Let $H$ be an $n\times n$ $\cal PT$-symmetric matrix and $\eta$ be the metric matrix of
$H$. Let $J$ and $S$ be matrices in Eqs (\ref{cano2}) and (\ref{cano2'}).
Then, there exist $n\times n$
invertible matrices $\Psi$, $\Xi$, $\Sigma$ and a $2n\times 2n$ Hermitian matrix $\tilde{H}$ such that for
$\tilde{\Psi}=\bpm\Psi\\ \Xi\epm$ and $\tilde{\Phi}=\bpm\Psi\\ \Sigma\epm$, the following equations hold,
\begin{eqnarray}
\tilde{\Phi}^\dag\tilde{\Psi}=S,~~~\tilde{\Phi}^\dag\tilde{H}\tilde{\Psi}=SJ.\label{23}
\end{eqnarray}
\end{thm}

\begin{proof}
As was discussed, there exist a matrix $\Psi'$ such that
$\Psi'^{-1}H\Psi'=J$ and $\Psi'^\dag\eta\Psi'=S$~\cite{Huang,Gohberg}.
Since $\Psi'^\dag\Psi'>0$, there always exits a positive number $c$ such that $c^2\Psi'^\dag \Psi' >I$. Set $\Psi=c\Psi'$.
Since $\Psi^\dag\Psi>I\geqslant S$, $\Psi^\dag\Psi-S$ is invertible.

Let $\Xi$ be an $n\times n$ invertible matrix. Taking $\Sigma=(\Xi^{-1})^\dag(S-\Psi^\dag\Psi)$,
$\eta=(\Psi^{-1})^\dag S\Psi^{-1}$, $H_1=\eta H$, $H_2=(\Psi^\dag)^{-1}(\Xi)^\dag$ and $H_4=-H_2^\dag\Psi\Xi^{-1}-(\Sigma^\dag)^{-1}\Psi^\dag H_2$,
one can directly verify that $\tilde{H}=\bpm H_1&H_2\\H_2^\dag&H_4\epm$ is Hermitian and Eq. (\ref{23}) holds.
\end{proof}

Theorem \ref{thm1} actually gives out the inner product structure of $H$ in a subspace.
Note that the matrix $\Psi$ in Theorem \ref{thm1} can be written as $\Psi=(\ket{\psi_1},\cdots,\ket{\psi_n})$, where the column vectors $\{\ket{\psi_i}\}$ form a linear basis of $\mathbb C^n$.
Similarly, $\Xi=(\ket{\xi_1},\cdots,\ket{\xi_n})$ and $\Sigma=(\ket{\sigma_1},\cdots,\ket{\sigma_n})$.
Correspondingly we have $\tilde{\Psi}=(\ket{\tilde{\psi}_1},\cdots,\ket{\tilde{\psi}_n})$ and $\tilde{\Phi}=(\ket{\tilde{\phi}_1},\cdots,\ket{\tilde{\phi}_n})$, where
$\ket{\tilde{\psi}_i}=\bpm\ket{\psi_i}\\ \ket{\xi_i}\epm$ and $\ket{\tilde{\phi}_i}=\bpm\ket{\psi_i}\\ \ket{\sigma_i}\epm$.
Moreover, $\tilde{\Phi} S=(\ket{\tilde{\mu}_1},\cdots,\ket{\tilde{\mu}_n})=(\ket{\tilde{\phi}_{s(1)}},\cdots,\ket{\tilde{\phi}_{s(n)}})$, where
$S$ is the permutation matrix in Theorem \ref{thm1}, and
$s$ is the permutation induced by $S$.
Similarly, we can write $\Psi S=(\ket{\mu_1},\cdots,\ket{\mu_n})$, where $\ket{\mu_i}=\ket{\psi_{s(i)}}$.
From the definition of $\ket{\tilde{\mu}_i}$, we have $\braket{\tilde{\mu}_i|\tilde{\psi}_j}=(S\tilde{\Phi}^\dag\tilde{\Psi})_{ij}$ and
$\braket{\tilde{\mu}_i|\tilde{H}|\tilde{\psi}_j}=(S\tilde{\Phi}^\dag\tilde{H}\tilde{\Psi})_{ij}$. According to Eq. (\ref{23}), we have
\begin{eqnarray}\label{ijiHj}
\braket{\tilde{\mu}_i|\tilde{\psi}_j}=\delta_{i,j},~~~\braket{\tilde{\mu}_i|\tilde{H}|\tilde{\psi}_j}=J_{i,j},
\end{eqnarray}
where $J_{i,j}$ is the $(i,j)$-th entry of $J$.

On the other hand, note that the metric matrix  $\eta$ of $H$ is $(\Psi^\dag)^{-1}S\Psi^{-1}$.
Thus we have the following relations between the Dirac and $\eta$-inner products 
\begin{eqnarray}
&&\braket{\tilde{\mu}_i|\tilde{\psi}_j}=\braket{\mu_i|\psi_j}_{\eta},\label{ij'}\\
&&\braket{\tilde{\mu}_i|\tilde{H}|\tilde{\psi}_j}=\braket{\mu_i|H|\psi_j}_{\eta},\label{iHj'}
\end{eqnarray}
where $\braket{\mu_i|H|\psi_j}_{\eta}=\braket{\mu_i|\eta H|\psi_j}$.
The results show that there exist two different basis with the same projections onto the subspace of the $\cal PT$-symmetric system, with respect to the $\eta$-inner product.
When confined to the subspace, the Hermitian Hamiltonian $\tilde{H}$ in large space has the same effect as a $\cal PT$-symmetric Hamiltonian $H$, in the sense of this $\eta$-inner product.


\noindent{\it Simulation of $\cal PT$-symmetric Hamiltonian systems}~~
To infer a quantum system is $\cal PT$-symmetric, it is sufficient to identify the Hamiltonian and its inner product structure.
In the weak measurement formalism, one starts by pre-selecting an initial state $\ket{\varphi_i}$.
The target system is coupled to the measurement apparatus, which is in a pointer state $\ket{P}$.
Usually, $\ket{P}=(2\pi\Delta^2)^{-\frac{1}{4}}exp(-\frac{Q^2}{4\Delta^2})$, a Gaussian state with $\Delta$ its standard
deviation. Let $A$ be an observable of the system and $M$ be that of the apparatus, conjugate to $Q$~\cite{AAV}.
The interaction Hamiltonian between the system and apparatus is $H_{int}=f(t)A\otimes M$, with interaction strength $g=\int f(t)dt$.
The state evolves as $\ket{\varphi_i}\otimes\ket{P}\rightarrow e^{-ig A\otimes M}\ket{\varphi_i}\otimes\ket{P}$.
Now if the system satisfies the weak condition that $g /\Delta$ is sufficiently small, then
for a post-selected state $\ket{\varphi_f}$ that $\braket{\varphi_f|\varphi_i}\neq 0$, one has
$\bra{\varphi_f}e^{-ig A\otimes M}\ket{\varphi_i}\ket{P}\approx\braket{\varphi_f|\varphi_i} e^{-ig\braket{A}_wM}\ket{P}
=\braket{\varphi_f|\varphi_i}(2\pi\Delta^2)^{-\frac{1}{4}}exp(-\frac{(Q-g\braket{A}_w)^2}{4\Delta^2})$,
where $\braket{A}_w=\frac{\braket{\varphi_f|A|\varphi_i}}{\braket{\varphi_f|\varphi_i}}$ is called the weak value.
That is, the state is shifted by $g\braket{A}_w$. Thus the weak value $\braket{A}_w$ can be read out experimentally, as a generalization of the eigenvalues in Von Neumann measurement~\cite{Dressel}.

From Eq. (\ref{ijiHj}), we have $\lambda_i=J_{i,i}=\braket{\tilde{\mu}_i|\tilde{H}|\tilde{\psi}_i}=\frac{\braket{\tilde{\mu}_i|\tilde{H}|\tilde{\psi}_i}}{\braket{\tilde{\mu}_i|\tilde{\psi}_i}}$.
Therefore, the eigenvalues of $H$ can be obtained via a weak measurement, by pre-selecting the vector $\ket{\tilde{\psi}_i}$ and post-selecting the vector $\ket{\tilde{\mu}_i}$.
This observation implies that one can use weak measurement to simulate the measurements on a $\cal PT$-symmetric system.

In conventional quantum mechanics, the expectation value of a Hermitian Hamiltonian $H_0=\sum_i \lambda_i\ket{u_i}\bra{u_i}$  with respect to a sate $\ket{\psi_0}=\sum_i d_i\ket{u_i}$ is given by the inner product $\braket{\psi_0|H_0|\psi_0}$. For a $\cal PT$-symmetric Hamiltonian system with the metric matrix $\eta$, the expectation value of a Hamiltonian $H$ with respect to a state $\ket{u}=\sum_i a_i\ket{\psi_i}$
is instead given by $\braket{u|H|u}_\eta$.
Given two vectors $\ket{v}=\sum_i b_i\ket{\mu_i}$ and $\ket{w}=\sum_i c_i\ket{\psi_i}$ of the $\cal PT$-symmetric system.
Let $\ket{\tilde{v}}=\sum_i b_i\ket{\tilde{\mu}_i}$ (unnormalized  for convenience) and $\ket{\tilde{w}}=\sum_i c_i\ket{\tilde{\psi}_i}$ be two vectors in the extended system.
It follows from Eq. (\ref{iHj'}) that $\braket{v|H|w}_\eta=\braket{\tilde{v}|\tilde{H}|\tilde{w}}$.
Assume that $\ket{u}$ satisfies the condition $\braket{u|u}_\eta=\pm 1$.
Now take two states $\ket{\tilde{u}_1}=\sum_i a_{s(i)}\ket{\tilde{\mu}_i}$ and $\ket{\tilde{u}_2}=\sum_i a_i\ket{\tilde{\psi}_i}$, whose projections to the
$\cal PT$-symmetric subspace are both $\ket{u}$. Then we have
\begin{equation}
\frac{\braket{u|H|u}_\eta}{\braket{u|u}_\eta}=\frac{\braket{\tilde{u}_1|\tilde{H}|\tilde{u}_2}}{\braket{\tilde{u}_1|\tilde{u}_2}}.
\end{equation}
Therefore, confined to the $\cal PT$-symmetric subspace, a weak measurement can completely describe the expectations of $H$.

In conventional quantum mechanics, when an eigenvalue is detected, the measured state collapses to the corresponding eigenstate.
However, the problem in $\cal PT$-symmetric system is subtle.
According to Eq. (\ref{ij'}), $\braket{\psi_i|\psi_i}_\eta\neq 0$ only if $i=s(i)$.
This observation makes it reasonable to assume that for any vector
$\ket{u}=\sum_i a_i\ket{\psi_i}$ satisfying $\braket{u|u}_\eta\neq 0$, if $a_i\neq 0$, then $a_{s(i)}\neq 0$.
That is, if $\braket{u|u}_\eta\neq 0$, its vector components of $\ket{\psi_i}$ and $\ket{\psi_{s(i)}}$ take zero or nonzero values simultaneously,
while the eigenvalues associated with $\psi_i$ and $\psi_{s(i)}$ are either equal or complex conjugations.
In this case, one can generalize the detection of an eigenvalue of $\lambda_i$ in conventional quantum mechanics to the following.
For $\ket{u}=\sum_i a_i\ket{\psi_i}$, if the value of
$$
\frac{a_i\overline{a_{s(i)}}\lambda_i+\overline{a_i}a_{s(i)}\overline{\lambda}_i}{a_i\overline{a_{s(i)}}+\overline{a_i}a_{s(i)}}
$$
is detected \cite{note 3}, the state $\ket{u}$ will collapse to
$$
\frac{a_i\ket{\psi_i}+a_{s(i)}\ket{\psi_{s(i)}}}{|a_i\overline{a_{s(i)}}+a_{s(i)}\overline{a_i}|^{1\over 2}}.
$$
Apparently, when $i=s(i)$, the state $\ket{u}$ will collapse to $\ket{\psi_i}$, similar to the case of conventional quantum mechanics.
Note that $i=s(i)$ only if the system is unbroken $\cal PT$-symmetric, for which
it is analogous to conventional quantum mechanics and such an analogy in state collapse is not unexpected.

By pre- and post-selecting the states, we see that the weak measurements can successfully simulate an arbitrary $\eta$-inner product.
Furthermore, when confined to the subspace, the measurement results actually extract the same information as a $\cal PT$-symmetric Hamiltonian system.
Such information help us eventually infer that the subsystem is $\cal PT$-symmetric.


\noindent{\it Discussions and conclusion}~~We further discuss the mechanism and physical implications related to the weak measurement paradigm, by comparing it with the embedding paradigm~\cite{Gunther-PRL,Huang}.
The essence of the embedding paradigm is to realize the evolution of a $\cal PT$-symmetric Hamiltonian, by evolving the state under the Hermitian Hamiltonian in the large space and then project it to the subspace.
The key to this paradigm can be mathematically described as follows~\cite{Huang}:
For a given $n\times n$ unbroken $\cal PT$-symmetric Hamiltonian $H$, find a $2n\times 2n$ Hermitian matrix $\tilde{H}$, $n\times n$
invertible matrices $\Psi$, $\Xi$ so that $\tilde{\Psi}^\dag\tilde{\Psi}=I$ and the following equations
\begin{eqnarray}\label{3''4''}
e^{-it\tilde{H}}\tilde{\Psi}=\tilde{\Psi}e^{-itJ},~~~e^{-itH}\Psi=\Psi e^{-itJ}
\end{eqnarray}
hold, where $\tilde{\Psi}=\bpm \Psi\\ \Xi \epm$.
The equations are actually equivalent to the following conditions~\cite{note 1}:
\begin{eqnarray}\label{2'3'4'}
\tilde{\Psi}^\dag\tilde{\Psi}=I,~~\tilde{H}\tilde{\Psi}=\tilde{\Psi}J,~~H\Psi=\Psi J.
\end{eqnarray}
Equation (\ref{3''4''}) ensures that the unitary evolution $\tilde{U}(t)$ gives the evolution $U(t)$ of a $\cal PT$-symmetric Hamiltonian $H$ in a subspace.
 In this sense, the embedding paradigm gives a natural way of simulation.
Nevertheless, in the broken $\cal PT$-symmetric case, the solutions do not exist~\cite{Huang}.
In fact, Eq. (\ref{23}) is mathematically a generalization of Eq. (\ref{2'3'4'}) \cite{note 2}.
Like the case of the embedding paradigm, it is natural to further require that $\tilde{\Phi}^\dag e^{-it\tilde{H}}\tilde{\Psi}=Se^{-itJ}$,
so that $e^{-it\tilde{H}}$ gives the same effect as $e^{-itH}$ in the subspace.
However, such a requirement cannot be satisfied for broken $\cal PT$-symmetry, which is obvious from the unboundedness of $Se^{-itJ}$.

However, consider sufficiently small time $t\in [0,\epsilon]$. We have
$\ket{\tilde{u}(t)}=e^{-it\tilde{H}}\ket{\tilde{u}}\approx(I-it\tilde{H})\ket{\tilde{u}}$. On the other hand,
$\ket{u(t)}=e^{-itH}\ket{u}\approx(I-itH)\ket{u}$.
Now equations Eqs. (\ref{ij'}) and (\ref{iHj'}) insure that when confined to the subspace,
$\ket{\tilde{u}(t)}$ is equivalent to $\ket{u(t)}$ in the sense of $\eta$-inner product \textcolor[rgb]{1.00,0.00,0.00}{(see Supplemental Material for an example)}.
This observation implies that $\cal PT$-symmetric quantum systems can be well approximated
in a sufficiently small time evolution,
by choosing two different sets of basis $\{\ket{\tilde{\phi}_i}\}$  and $\{\ket{\tilde{\psi}_i}\}$ with the same components in the subspace, which
can be realized by weak measurement. Here instead of the small time interval,
the weak condition that $g/\Delta$ is sufficiently small
ensures the approximation.
The weak measurement paradigm can be viewed as a generalization of the embedding paradigm, due to the fact that Eq. (\ref{2'3'4'}) is a special case of
Eq. (\ref{23}) in the $\cal PT$-symmetric unbroken case.
Hence, the Hamiltonian $\tilde{H}$ in the embedding paradigm can also be utilized in the weak measurement approach,
although the embedding paradigm itself does not work.
Comparing our approach with that in~\cite{Pati}, where one obtains the expected value of a Hamiltonian in the Dirac inner product
by using the polar decomposition, our method lays emphasis on the properties of a $\cal PT$-symmetric Hamiltonian with respect to the $\eta$-inner product.

In summary, we have proposed a weak measurement paradigm to investigate the behaviors of broken $\cal PT$-symmetric Hamiltonian systems.
By embedding the $\cal PT$-symmetric system into a large Hermitian system and utilizing weak measurements, we have shown how a $\cal PT$-symmetric Hamiltonian can be simulated.
Our paradigm may shine new light on the study of $\cal {PT}$-symmetric quantum mechanics and its physical implications and applications.

\begin{acknowledgments}
This work is supported by National Natural Science Foundation of China (11171301, 11571307, 11690032, 61490711, 61877054 and 11675113), National Key R\&D Program of China under Grant No. 2018 YFA0306202 and the NSF of Beijing under Grant No. KZ201810028042.
\end{acknowledgments}


\begin{thebibliography}{99}%
\bibitem{Bender98} C. M. Bender and S. Boettcher, Phys. Rev. Lett. 80, 5243 (1998).


\bibitem{El-OL} R. El-Ganainy, K. Makris, D. Christodoulides, and Z. H. Musslimani,
Opt. Lett. 32, 2632 (2007).

\bibitem{Makris-PRL}
K. G. Makris, R. El-Ganainy, D. Christodoulides, and Z. H. Musslimani,
Phys. Rev. Lett. 100, 103904 (2008).

\bibitem{Guo-PRL} A. Guo, G. Salamo, D. Duchesne, R. Morandotti, M. Volatier-Ravat, V. Aimez, G. Siviloglou, and D. Christodoulides,
Phys. Rev. Lett. 103, 093902 (2009).

\bibitem{Ruter-NP} C. E. Ruter, K. G. Makris, R. El-Ganainy, D. N. Christodoulides, M. Segev, and D. Kip,
Nat. Phys. 6, 192 (2010).


\bibitem{Chang-NP} L. Chang, X. Jiang, S. Hua, C. Yang, J. Wen, L. Jiang, G. Li, G. Wang, and M. Xiao,
Nat. Photon. 8, 524 (2014).

\bibitem{Tang-NP} J.-S. Tang, Y.-T. Wang, S. Yu, D.-Y. He, J.-S. Xu, B.-H. Liu, G. Chen, Y.-N. Sun, K. Sun, Y.-J. Han, C.-F. Li, and G.-C. Guo,
Nat. Photon. 10, 642 (2016).

\bibitem{Lee14} Y.-C. Lee, M.-H. Hsieh, S. T. Flammia, and R.-K. Lee, Phys. Rev. Lett. 112, 130404 (2014).

\bibitem{Dorey} P. Dorey, C. Dunning, and R. Tateo,
J. Phys. A: Math. Theor. 34, 5679 (2001); {\it ibid}  40, R205 (2007).



\bibitem{Gunther-PRA} U. G{\"u}nther and B. F. Samsonov,
Phys. Rev. A 78, 042115 (2008).





\bibitem{Bender-17} C. M. Bender, D. C. Brody, and M. P. Muller,
Phys. Rev. Lett. 118, 130201 (2017).





\bibitem{Mostafazadeh} A. Mostafazadeh, J. Math. Phys. 43, 205, 2814, and 3944 (2002).


\bibitem{Mostafazadeh-Geom} A. Mostafazadeh,
Int. J. Geom. Meth. Mod. Phys. 7, 1191 (2010).

\bibitem{JPA}   M. Huang, R.-K. Lee and J. Wu,
 J. Phys. A: Math. Theor. 51, 414004 (2018).


\bibitem{Gunther-PRL} U. G{\"u}nther and B. F. Samsonov,
Phys. Rev. Lett. 101, 230404 (2008).

\bibitem{Huang} M. Huang, A. Kumar and J. Wu,
Phys. Lett. A 382, 2578 (2018).

\bibitem{Ueda} K. Kawabata,  Y. Ashida,  and M. Ueda,
Phys. Rev. Lett. 119, 190401 (2017).


\bibitem{AAV} Y. Aharonov, D. Z. Albert, and L. Vaidman,
Phys. Rev. Lett. 60, 1351 (1988).

\bibitem{Aharonov}  Y. Aharonov, A. Botero, S. Popescu, B. Reznik, and J. Tollaksen, 
Phys. Lett. A 301, 130 (2002).

\bibitem{Resch} K. J. Resch, J. S. Lundeen, and A. M. Steinberg, 
Phys. Lett. A 324, 125 (2004).

\bibitem{Palacios} A. Palacios-Laloy, F. Mallet, F. Nguyen, P. Bertet, D. Vion, D. Esteve, and A. N. Korotkov, 
Nat. Phys. 6, 442 (2010).

\bibitem{Hosten} O. Hosten and P. Kwiat,
Science 319, 787 (2008).

\bibitem{Sjogvist} E. Sjoqvist,
Phys. Lett. A 359, 187 (2006).

\bibitem{Kobayashi} H. Kobayashi, S. Tamate, T. Nakanishi, K. Sugiyama, and M. Kitano,
Phys. Rev. A 81, 012104 (2010).

\bibitem{Lijian}  L. Zhang,  A. Datta, and I. A. Walmsley,
Phys. Rev. Lett. 114, 210801 (2015)

\bibitem{Pati} A. K. Pati, U. Singh, and U. Sinha, 
Phys. Rev. A 92, 052120 (2015).

\bibitem{Vaidman} L. Vaidman, A. Ben-Israel, J. Dziewior, L. Knips, M. Wei{\ss}l, J. Meinecke, C. Schwemmer, R. Ber, and H. Weinfurter,
Phys. Rev. A 96, 032114 (2017).

\bibitem{Lundeen} J. S. Lundeen, B. Sutherland, A. Patel, C. Stewart, and C. Bamber, 
Nature 474, 188 (2011).


\bibitem{Starling-2010b} D. J. Starling, P. B. Dixon, A. N. Jordan, and J. C. Howell, 
Phys. Rev. A 82, 063822 (2010).


\bibitem{Brunner} N. Brunner and C. Simon, 
Phys. Rev. Lett. 105, 010405 (2010).






\bibitem{note} In this paper $\cal PT$-symmetry is the synonym of pseudo-Hermitian, between which we will not distinguish.


\bibitem{Deng} J.-w. Deng, U. Guenther, and Q.-h. Wang, 
arXiv:1212.1861 (2012).

\bibitem{Mannheim} P. D. Mannheim, 
Phil. Trans. Royal Soc. London A: Math. Phys. Eng. Sci. 371, 20120060 (2013).

\bibitem{Horn} R. A. Horn and C. R. Johnson, {\it Matrix analysis,} (Cambridge University, 2012).





\bibitem{Gohberg}  I. Gohberg, P. Lancaster, and L. Rodman, {\it Matrices and indefinite scalar products,} vol. 8 (1983). 

\bibitem{Dressel} J. Dressel, M. Malik, F. M. Miatto, A. N. Jordan, and R. W. Boyd, 
Rev. Mod. Phys. 86, 307 (2014).

\bibitem{note 3}
Actually, this value is $\frac{\braket{a_i\psi_i+a_{s(i)}\psi_{s(i)}|H|a_i\psi_i+a_{s(i)}\psi_{s(i)}}_\eta}{a_i\overline{a_{s(i)}}+\overline{a_i}a_{s(i)}}$, which reduces to $\braket{\psi_i|H|\psi_i}_\eta=\lambda_i$ if $i=s(i)$ (only one vector $\psi_i$ considered).

\bibitem{note 1} Equation
(\ref{3''4''}) is actually the matrix version of the embedding in~\cite{Huang}.
Denote $\Xi=\tau \Psi$. Then (\ref{2'3'4'}) reduces to $H_1+H_2\tau=H$ and $H_2^\dag+H_4\tau=\tau H$, which gives the equivalent description of the embedding property.
A concrete solution to (\ref{3''4''}) can also be found in \cite{Ueda}.

\bibitem{note 2}
When unbroken $\cal PT$-symmetric, it is always possible to take $\Xi=\Sigma$ and $S=I$, Eq. (\ref{2'3'4'}) is a special case of Eq. (\ref{23}).
%


\end{thebibliography}

\begin{thebibliography}{99}
\bibitem{Bender07} C. M. Bender, D. C. Brody, H. F. Jones, and B. K. Meister, Phys. Rev. Lett. 98, 040403 (2007).

\bibitem{Gunther-PRL} U. G{\"u}nther and B. F. Samsonov,
Phys. Rev. Lett. 101, 230404 (2008).

\bibitem{Ueda} K. Kawabata,  Y. Ashida,  and M. Ueda,
Phys. Rev. Lett. 119, 190401 (2017).

\bibitem{Huang} M. Huang, A. Kumar, and J. Wu,
Phys. Lett. A 382, 2578 (2018).

\end{thebibliography}

\newpage

\section{Supplemental Material:}
\subsection{An example}
To illustrate the validity of our theoretic results, an example is given based on the two dimensional model proposed by Bender et al.~\cite{Bender07}:
$$
H=\begin{bmatrix}
re^{i\theta} &s\\
s&re^{-i\theta}
\end{bmatrix},~~~~~~P=\begin{bmatrix}
0 & 1\\
1 & 0
\end{bmatrix},
~~~~~~T=\begin{bmatrix}
1 & 0\\
0 & 1
\end{bmatrix}.
$$
Here, as $HPT=PT\overline{H}$, the Hamiltonian $H$ is $\cal PT$-symmetric.~In particular, when $\Delta=s^2-r^2\sin^2\theta< 0$, $H$ is broken $\cal PT$-symmetric. The corresponding eigenvalues and eigenvectors (without
 normalization)  are:
$$\lambda_1=r\cos\theta+i\sqrt{-\Delta},~~~\lambda_2=r\cos\theta-i\sqrt{-\Delta}.$$
$$\psi_1=\begin{bmatrix}i(\sqrt{-\Delta}+r\sin\theta)\\s \end{bmatrix}, \psi_2=\begin{bmatrix}-s \\ i(\sqrt{-\Delta}+r\sin\theta)\end{bmatrix}.$$

Then, by denoting the eigenvectors in the matrix form, we have:
\begin{eqnarray}
\Psi=[\psi_1,\psi_2]=\begin{bmatrix}i(\sqrt{-\Delta}+r\sin\theta) & -s\\s &i(\sqrt{-\Delta}+r\sin\theta)\end{bmatrix},
\end{eqnarray}
It can be verified that $\Psi^{-1}H\Psi=J$ and $\Psi^\dag\eta\Psi=S$, where
\begin{eqnarray}
J=\begin{bmatrix}r\cos\theta+i\sqrt{-\Delta}&0\\0& r\cos\theta-i\sqrt{-\Delta} \end{bmatrix},~~~S=\begin{bmatrix}
0 & 1\\
1 & 0
\end{bmatrix}.
\end{eqnarray}


Now, with the short-handed notations,
$$u=\sqrt{-\Delta}+r\sin\theta,~a=2r\sin\theta.$$
%
we have
\begin{eqnarray}
&&\Psi=\begin{bmatrix}iu & -s\\s &iu \end{bmatrix},~\Psi^{-1}=\frac{1}{s^2-u^2}\begin{bmatrix}iu & s\\-s &iu \end{bmatrix},\\
&&S-\Psi^\dag\Psi=\begin{bmatrix}-au&1-2sui\\1+2sui&-au\end{bmatrix},\\
&&det(S-\Psi^\dag\Psi)=-4\Delta u^2-1.
\end{eqnarray}
For simplicity, we also assume $-4\Delta u^2-1\neq 0$. Otherwise, as showed in the proof of Theorem $1$, we can take a constant value $c$ such that $S-c^2\Psi^\dag\Psi$ is invertible,  i.e., with $c\Psi$ instead of $\Psi$. Now
$$(S-\Psi^\dag\Psi)^{-1}=\frac{1}{-4\Delta u^2-1}\begin{bmatrix}-au & -1+2sui \\-1-2sui &-au\end{bmatrix}.$$

To introduce our simulating scenario, we take $\Psi=\Xi$ for convenience, as $\Xi$ is arbitrary.
By using the construction in Theorem $1$, one can have $\Sigma=(\Xi^{-1})^\dag(S-\Psi^\dag\Psi)$,
$\eta=(\Psi^{-1})^\dag S\Psi^{-1}$, $H_1=\eta H$, $H_2=(\Psi^\dag)^{-1}(\Xi)^\dag$ and $H_4=-H_2^\dag\Psi\Xi^{-1}-(\Sigma^\dag)^{-1}\Psi^\dag H_2
=-I-\Xi(S-\Psi^\dag\Psi)^{-1}\Xi^\dag$.

Then, direct calculations give us
\begin{eqnarray}
&&\tilde{H}=\begin{bmatrix}A_1&A_2&1&0\\A_3&A_4&0&1 \\ 1&0& -1-KB_1&-KB_2\\ 0&1&-KB_3&-1-KB_4\end{bmatrix},\\
&&\tilde{\Psi}=\begin{bmatrix}iu&-s\\s&iu\\iu&-s\\s&iu\end{bmatrix},\\
&&\tilde{\Phi}^\dag=\begin{bmatrix}-iu&s&iu-K_2s &iK_2 u-s\\-s&-iu&iK_2 u+s&iu+K_2s\end{bmatrix},
\end{eqnarray}
with the notations
\begin{eqnarray*}
&&K=\frac{1}{-4\Delta u^2-1},~~K_2=\frac{1}{s^2-u^2},\\
&&A_1=\frac{s}{u^2-s^2},~~~A_2=\frac{re^{-i\theta}}{u^2-s^2},\\
&&A_3=\frac{re^{i\theta}}{u^2-s^2},~~~A_4=\frac{s}{u^2-s^2},\\
&&B_1=B_4=-(u^2-s^2)^2,\\
&&B_2=B_3=s^2-u^2.
\end{eqnarray*}

Based on these solutions, it can be easily verified that $\tilde{\Phi}^\dag\tilde{\Psi}=S,\tilde{\Phi}^\dag\tilde{H}\tilde{\Psi}=SJ$, such that
\begin{eqnarray}
\braket{\tilde{\phi}_i,e^{-it\tilde{H}}\tilde{\psi}_j}\approx \tilde{\phi}_i^\dag(I-it\tilde{H})\tilde{\psi}_j
=\psi_i^\dag\eta(I-itH)\psi_j\approx\braket{\psi_i,e^{-itH}\psi_j}_\eta.\nonumber \\
\label{approx}
\end{eqnarray}

Thus, under the $\eta$-inner product, the reduced system resembles a broken $\cal PT$-symmetric one.
In order to illustrate the validity of our simulating paradigm, we introduce four parameters defined below:
\begin{eqnarray}
&&Z_{11}=|\braket{\tilde{\phi}_1,e^{-it\tilde{H}}\tilde{\psi}_1}|,\\
&&Z_{22}=|\braket{\tilde{\phi}_2,e^{-it\tilde{H}}\tilde{\psi}_2}|,\\
&&Z_{12}=|\braket{\tilde{\phi}_1,e^{-it\tilde{H}}\tilde{\psi}_2}-\braket{\psi_1,e^{-itH}\psi_2}_\eta||\braket{\psi_1,e^{-itH}\psi_2}_\eta|^{-1},\nonumber\\
&&\\
&&Z_{21}=|\braket{\tilde{\phi}_2,e^{-it\tilde{H}}\tilde{\psi}_1}-\braket{\psi_2,e^{-itH}\psi_1}_\eta||\braket{\psi_2,e^{-itH}\psi_1}_\eta|^{-1}.\nonumber\\
\end{eqnarray}

The reason $Z_{11}$ an $Z_{22}$ have different forms from $Z_{12}$ and  $Z_{21}$ is that   $\braket{\psi_1,e^{-itH}\psi_1}_\eta=\braket{\psi_2,e^{-itH}\psi_2}_\eta=0$, but
$\braket{\psi_1,e^{-itH}\psi_2}_\eta \neq 0, \quad \braket{\psi_2,e^{-itH}\psi_1}_\eta\neq 0.
$
With the definitions above, apparently, $Z_{ij}$  reflects the difference between $\braket{\tilde{\phi}_i,e^{-it\tilde{H}}\tilde{\psi}_j}$ and $\braket{\psi_i,e^{-itH}\psi_j}_\eta$, as shown in FIG. S1.

\begin{widetext}
\centering
\begin{figure}[ht]
\includegraphics[width=16.0cm]{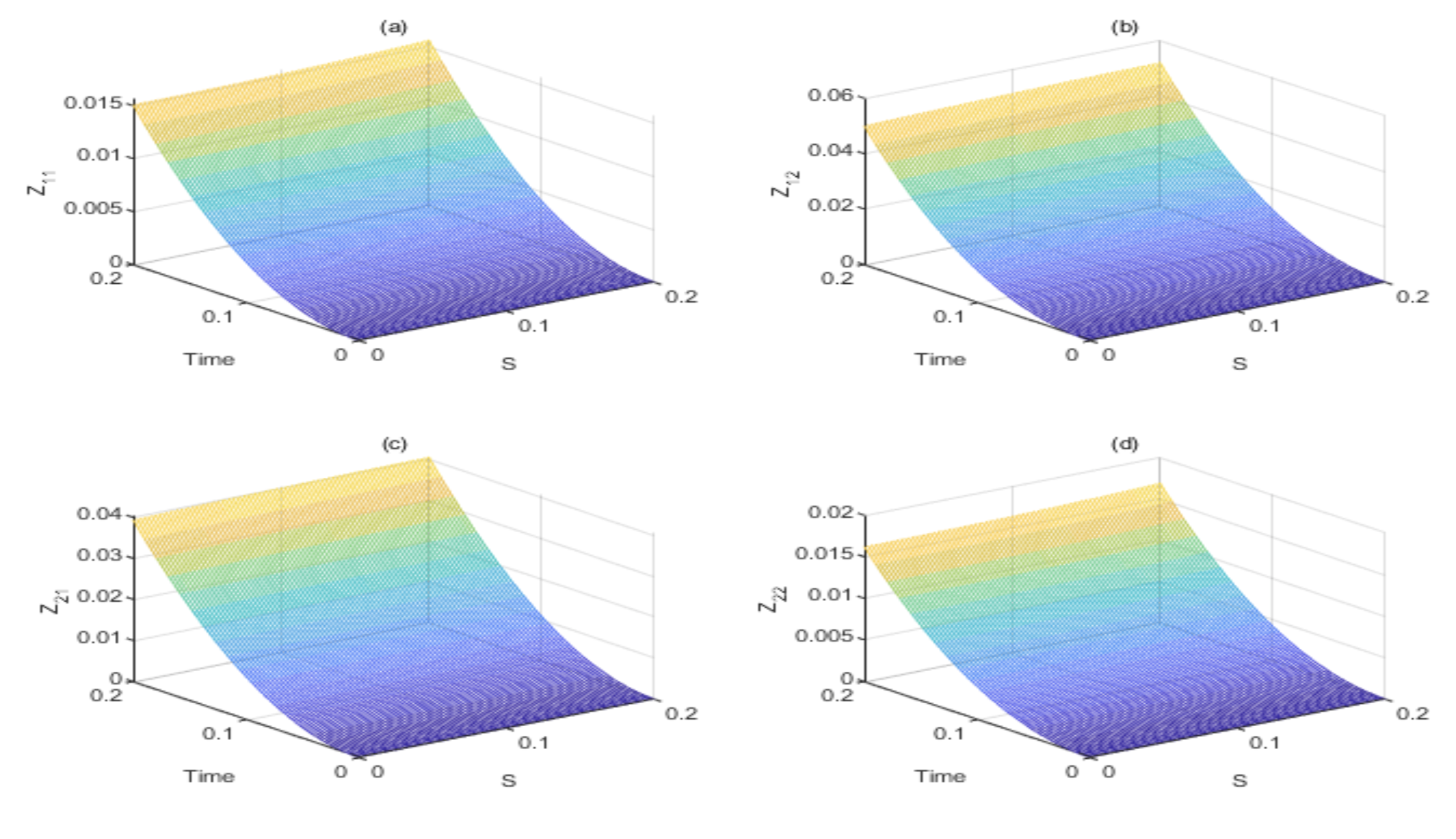}
\caption{Direct calculation on the parameters $Z_{ij}$. The corresponding values of $Z_{ij}$ given in Eqs. (S10-S13) are shown in (a-d), respectively, for the range of $r=\sqrt{2},\theta=\frac{\pi}{4}, t\in[0,0.2],s\in[0,0.2]$.}
\end{figure}
\end{widetext}

In the range of $r=\sqrt{2},\theta=\frac{\pi}{4}, t\in[0,0.2],s\in[0,0.2]$, the differences between $\braket{\tilde{\phi}_i,e^{-it\tilde{H}}\tilde{\psi}_i}$ and $\braket{\psi_i,e^{-itH}\psi_i}_\eta$ are less than $2 \times 10^{-2}$; while  the relative differences between $\braket{\tilde{\phi}_i,e^{-it\tilde{H}}\tilde{\psi}_j}$ and $\braket{\psi_i,e^{-itH}\psi_j}_\eta$ are less than $6 \times 10^{-2}$.

In addition, when $t\rightarrow 0$, $Z_{ij}$ and thus $\braket{\tilde{\phi}_i,e^{-it\tilde{H}}\tilde{\psi}_j}-\braket{\psi_i,e^{-itH}\psi_j}_\eta$, tend to zero. This means that Eq.~(S9) is valid for a sufficiently small time interval $t$, which supports our theoretical conclusion. We want to emphasize that $e^{-it\tilde{H}}$ behaves like a
broken $\cal PT$-symmetric evolution under the $\eta$-inner product, but not under the standard Dirac inner product.
Hence in this case,   the projection of $e^{-it\tilde{H}}\tilde{\psi}_i$ is not expected to be the same as that of  $e^{-itH}\psi_i$.

Moreover, our theorem gives the same results for unbroken $\cal PT$-symmetry.  When $\cal PT$-symmetry is unbroken, then $S=I$, $\eta=(\Psi^{-1})^{\dag}\Psi^{-1}>0$, $J$ is diagonal,  resulting in Eq.~(9)  being just a special case of our Theorem $1$.  Apparently, Eq.~(9) implies that the projection of $e^{-it\tilde{H}}\tilde{\psi}_i$ is numerically equal to $e^{-itH}\psi_i$. Hence the embedding paradigms illustrated in Refs.~\cite{Gunther-PRL,Ueda,Huang} are also included in our method, although in those papers the $\eta$-inner product and measurements are not considered on purpose.

With the help of  the analogy between Dirac inner product and $\eta$-inner product of unbroken $\cal PT$-symmetry, the example illustrated in Ref.~\cite{Gunther-PRL} can be viewed as a proof for  our paradigm in the unbroken $\cal PT$-symmetry. Explicitly, one can verify that Eq.~(3) holds for the construction given below:
\begin{widetext}
\begin{eqnarray*}
&&H=\begin{bmatrix}
E_0+is\sin\theta &s\\
s&E_0-is\sin\theta
\end{bmatrix},~S=\begin{bmatrix}
1 &0\\
0 &1
\end{bmatrix},~
J=\begin{bmatrix}
E_0+s\cos\theta &0\\
0&E_0-s\cos\theta
\end{bmatrix},\\
&&\tilde{H}=
\begin{bmatrix}E_0&s\cos^2\theta&is\cos\theta\sin\theta&0\\s\cos^2\theta&E_0&0&-is\cos\theta\sin\theta
 \\ -is\cos\theta\sin\theta&0& E_0&s\cos^2\theta\\ 0&is\cos\theta\sin\theta&s\cos^2\theta&E_0\end{bmatrix},\\
&&\tilde{\Psi}=\tilde{\Phi}=
\begin{bmatrix}
\frac{e^{\frac{i\theta}{2}}}{2}&\frac{ie^{\frac{-i\theta}{2}}}{2}\\
\frac{e^{-\frac{i\theta}{2}}}{2}&-\frac{ie^{\frac{i\theta}{2}}}{2}\\
\frac{e^{-\frac{i\theta}{2}}}{2}&\frac{ie^{\frac{i\theta}{2}}}{2}\\
\frac{e^{\frac{i\theta}{2}}}{2}&-\frac{ie^{-\frac{i\theta}{2}}}{2}
\end{bmatrix}.
\end{eqnarray*}
\end{widetext}

\end{document}